\documentclass{acmsmall-ec13}
\usepackage{fancyhdr}
\usepackage[ruled]{algorithm2e}

\SetAlFnt{\small}
\SetAlCapFnt{\small}
\SetAlCapNameFnt{\small}
\SetAlCapHSkip{0pt}
\IncMargin{-\parindent}

\usepackage{amsmath}

\def\E{\mathbb{E}}

\def\MP#1{\textbf{[[#1]]}}

\SetAlFnt{\small}
\SetAlCapFnt{\small}
\SetAlCapNameFnt{\small}
\SetAlCapHSkip{0pt}
\IncMargin{-\parindent}

\acmVolume{X}
\acmNumber{X}
\acmArticle{X}
\acmYear{2013}
\acmMonth{6}

\newtheorem{observation}{Observation}[section]

\def\pr {\mathrm{Pr}}
\long\def\hide#1{{}}

\begin{document}

\lhead{} 
\rhead{} 
\lfoot{} 
\cfoot{} 
\rfoot{} 
\chead{\bfseries Proceedings Article}  


\title{Sybil-proof Mechanisms in Query Incentive Networks}
\author{WEI CHEN
\affil{Microsoft Research Asia, \texttt{weic@microsoft.com}}
YAJUN WANG
\affil{Microsoft Research Asia, \texttt{yajunw@microsoft.com}}
DONGXIAO YU
\affil{The University of Hong Kong, \texttt{dxyu@cs.hku.hk}}
LI ZHANG
\affil{Microsoft Research Silicon Valley Lab, \texttt{lzha@microsoft.com}}}

%

%



%

\begin{abstract}
In this paper, we study incentive mechanisms for retrieving
information from networked agents.  Following the model in
\cite{Kleinberg2005}, the agents are represented as nodes in an
infinite tree, which is generated by a random branching process. A
query is issued by the root, and each node possesses an answer with
an independent probability $p=1/n$. Further, each node in the tree
acts strategically to maximize its own payoff.  In order to encourage
the agents to participate in the information acquisition process, an
incentive mechanism is needed to reward agents who provide the
information as well as agents who help to facilitate such acquisition.

We focus on designing efficient sybil-proof incentive mechanisms, i.e., which are robust to fake identity attacks.
We propose a family of mechanisms, called the {\em direct referral}
(DR) mechanisms, which allocate most reward to the information holder
as well as its direct parent (or direct referral).  We show that, when
designed properly, the direct referral mechanism is sybil-proof and
efficient. In particular, we show that we may achieve an expected cost
of $O(h^2)$ for propagating the query down $h$ levels for any
branching factor $b>1$.  This result exponentially improves on previous work when requiring to find an answer with high probability.
When the underlying network is a deterministic chain, our mechanism is optimal under some mild assumptions.
In addition, due to its simple reward structure, the DR
mechanism might have good chance to be adopted in practice.


\end{abstract}

\begin{bottomstuff}
\end{bottomstuff}
\maketitle
\thispagestyle{fancy}

\section{Introduction}

%

Many information systems, e.g., peer-to-peer networks or social networks, are designed such that queries are answered by networked
agents instead of a centralized authority.  In such a system, the
query propagates in the network with the hope that it will eventually
reach some agents which hold (and return) an answer.  For such query
models, it is important to design an incentive mechanism to encourage
the query propagation and the return of the answer.  In addition, we
would like the mechanism to be \emph{efficient}, with low expected
cost for the root, and \emph{sybil-proof}, discouraging the
agents to disrupt or delay the query process by producing fake
identities.  In this paper, we will present a family of mechanisms
which can achieve these two goals simultaneously.

We 
{mainly} follow the query incentive network model invented by
\citeN{Kleinberg2005}. Under the model, each agent is represented as a
node in a fixed infinite $d$-ary tree. A query is issued by the root,
where each node in the tree may have an answer with a fixed
probability $p$. Information, query or answer, can propagate along the
edges in the tree which are ``turned on'' according to a random branching
process.  Each node, with a local view of the tree, can decide if it
continues to propagate the query to its children and to forward back an
answer to its parent.  The nodes are self-interested and risk-neutral so they choose
the actions that maximize their expected payoff.
In the model, there is a fixed unit cost in forwarding
a selected answer along each edge. On the other hand, it is free to propagate the queries.

In~\cite{Kleinberg2005}, the authors considered incentive mechanisms
in the form of \emph{fixed-payment contract}, where each node offers a fixed
amount of reward, which is in turn a part of the reward promised by its parent, to its
children, under the condition that the child propagates back an
answer (and accepted by the root).  In the paper, the authors obtain
the lowest possible cost (or reward) needed for such mechanisms to
retrieve an answer with a constant probability.  The cost depends on
the \emph{rarity} of the answer, defined as $n=1/p$, and the branching
factor $b$, the expected number of edges to children that are ``turned on'' by the
random process. The paper showed a phase transition phenomenon at $b=2$: when
$b>2$, the mechanism can achieve low cost of
$O(\log n)$; when $b<2$, the cost explodes to $n^{\Omega(1)}$, which is exponential to the
number of levels needed to explore.

The work of~\citeN{Kleinberg2005} motivated many subsequent works.  In
particular, \citeN{Arcaute2007} extended the results to more general
random processes, and observed the similar phase transition phenomenon.
\citeN{Cebrian2012} analyzed a different type of mechanism, called
\emph{split contract}, which models a successful scheme used by the
winner~\cite{Pickard2011} of the DARPA Network Challenge (also known
as Red Balloon Challenge)~\cite{DARPA2009}.  Roughly speaking, in the
split contract mechanism, the answer holder receives the specified
reward, while each ``referral'' on the path to the root receives a
fraction of its child's reward. \citeN{Cebrian2012}
showed that the split contract can achieve low cost even when $b<2$.
However, unlike the fixed-payment contract scheme, the split contract is not sybil-proof, as one can
produce fake identities in the tree to obtain higher expected payoff.

In our paper, we propose a new family of mechanisms which distribute most reward
to, in addition to the agent who provides the answer, the direct referral,
i.e., its parent. We call such query incentive mechanisms as {\em
Direct Referral (DR)}
mechanisms.  We show that the direct
referral mechanism, when designed properly, can discourage sybils,
i.e., it is not to the {agents'} interest to create fake identities, as
well as obtain a low cost to retrieve an answer for any $b>1$.
In particular, the scheme still has a low cost when the success probability
is close to $1-\zeta$, e.g., with probability $1-\zeta - O(1/n)$, where $\zeta$ is the extinction probability of the branching process.
Both fixed-payment contracts and split contracts incur high cost when the success probability approaches $1-\zeta$.

The following two theorems summarize the main results.
\def\thmChain
{
If the underlying branching process is a deterministic chain, there exists a sybil-proof direct referral incentive mechanism with expected cost $O(nh^2)$, where $h$ is the desired level of agents the root wants to propagate the query and $n$ is the answer rarity.
}
\begin{theorem}
\label{thm:chain}
\thmChain
\end{theorem}
Furthermore, the direct referral mechanism is optimal on the chain among all the sybil-proof query incentive mechanisms which satisfy some mild assumptions (Section~\ref{subsec:opt}).


\def\thmTree
{
For any branching process with branching factor $b>1$, there is a constant $p_b > 0$ such that for any answer rarity $n$ with  $0<1/n<p_b$, there exists a sybil-proof direct referral query incentive mechanism with expected cost $O(h^2)$, where $h$ is the desired level of agents the root wants to propagate the query.
}
\begin{theorem}
\label{thm:tree}
\thmTree
\end{theorem}


Notice that the above bound holds for $b>1$ so the direct referral
mechanism, compared to the fixed-payment contract mechanism, can achieve low cost
for a larger family of branching processes. In addition, the cost (when
$b>1$) has a polynomial dependence on the level $h$ instead of the
rarity $n$. Therefore, in the case of retrieving an answer with high
probability, i.e., a probability which is close to the
extinction probability of the branching process, the expected cost is
still polynomial in $h$, rather than polynomial in $n$ (in this case, $n$ is exponential in $h$).  In contrast,
query incentive networks with either fixed-payment
contracts~\cite{Arcaute2007} or split contracts~\cite{Cebrian2012}
have cost polynomial of $n$ in the high probability case.



The direct referral mechanism has the natural structure of
rewarding the answer holder and the direct referral while the others
only receive minimum compensation for routing the answer.  Such
simplicity might be highly desirable for practical adoption.  On the
other hand, despite the simplicity of the direct referral scheme, we
show that the mechanism can be quite robust, i.e., sybil-proof, and
efficient. Actually, in the case of the infinite chain, we can show
that such scheme is {\em optimal} among all the sybil-proof
mechanisms with some mild assumptions.

Intuitively, by rewarding the direct referral, the DR mechanism
encourages any agent who does not have an answer to propagate the query
as there is a chance that one of its children may have an answer (and
get selected by the root) so the agent can win the direct referral
reward. In addition, for any agent, no matter how many sybils it
creates, at most one of them receives the ``direct referral''
reward. Therefore, if we design the direct referral rewards such that
they decrease {rapidly} enough as the depth increases, the potential gain
of the sybils may be offset by the gap between the direct referral
rewards {in the two different levels.}
Of course, the gap cannot be too large for
it would increase the expected cost for the root.  Indeed, we will
show that with mildly decreasing direct referral rewards, we can
achieve both sybil-proofness and low cost at the same time.

\subsection{Related Work}

The model of query incentive networks was introduced by Kleinberg and Raghavan~\citeyear{Kleinberg2005},
where they considered a simple branching process in an underlying $d$-ary tree, i.e., each edge exists with an independent probability $\frac{b}{d}$ with branching factor $b>1$.
Kleinberg and Raghavan observed an interesting phase-transition phenomenon with fixed-payment contracts. Specifically, when $b>2$, in order to retrieve the answer with constant probability, the reward needed for the root to offer is $O(\log n)$ which is asymptotically optimal.  However, when $1<b<2$, the cost grows to $n^{\Theta(1)}$, i.e., the root needs to pay a reward that is exponentially larger than the expected distance for finding an answer in this case.

\citeN{Arcaute2007} generalized the simple branching process in~\cite{Kleinberg2005} to an arbitrary GW branching process in which the number of children of a node is determined by a fixed offspring distribution. They observed that the 
{phase transition phenomenon}
at branching factor $b=2$ still exists in this general case. Furthermore, they also observed that when it requires to find the answer with high probability, e.g., with probability $1-\zeta-\frac{1}{n}$, where $\zeta$ is the extinction probability of the branching process, the phase transition phenomenon vanishes. In particular, for any branching process with branching factor $b>1$, to retrieve an answer with high probability, the required reward is $n^{\Theta(1)}$. They also showed that in a deterministic chain ($b=1$), the cost of the root is $\Omega(n!)$ for finding an answer with constant probability. 

Kota and Narahari~\citeyear{Kota2010} analyzed the reward for such fixed-payment contracts when the degree distribution follows power-law.
Dikshit and Yadati~\citeyear{Dikshit2009} considered the 
{quality of the answers} in query incentive networks. 
Both models exhibit similar 
phase transition phenomena
at 
branching factor $b=2$ found in~\cite{Arcaute2007,Kleinberg2005}.

\citeN{Cebrian2012} presented a split contract based mechanism motivated by the success of the winning team of the DARPA challenge. In this mechanism, the root provides a reward $r$ to the answer holder. A ``shortest path'' based answer selection scheme is adopted in case of multiple reachable answers. See Section~\ref{sec:answer selection}.
Each node in the path to the selected answer will receive a fraction of the reward received by its child.
With the split contract mechanism, it is shown that the 
{phase transition phenomenon}
vanishes at $b=2$.
In particular, for any GW branching process with branching factor $b>1$, the cost to retrieve an answer with constant probability is $O(\log n)$ in a Nash equilibrium which is asymptotically optimal. Therefore, split contract based query incentive networks are more efficient. On the other hand, if we want to retrieve an answer with high probability, 
split contracts also need reward of $n^{\Omega(1)}$. 

In the previous studies on the query incentive networks, generating fake identities is not part of the agents' strategy. In other words, sybil-proofness is not explicitly explored.
In fact, we show that the fixed contract based mechanisms are ``sybil-proof'' while split contract mechanisms are clearly not.


How to prevent sybil attacks~\cite{Douceur2002} has been studied in many aspects of computer networks. \citeN{Babaioff2012} presented a sybil-proof scheme for the Bitcoin system. There are several major differences between the Bitcoin system and a query incentive network. Most notably, instead of a branching process, the network in the Bitcoin system can be intentionally constructed, which gives the mechanism designer an additional freedom to address the sybil-proofness. Therefore, the results in~\cite{Babaioff2012} cannot be directly applied to a query incentive network. They adopted the iterated removal of dominated strategy, which is a stronger solution concept than the Nash equilibrium used in this paper.


Sybil-proofness mechanisms for multi-level marketing is studied in~\cite{Emek2011,Drucker2012}. In the multi-level marketing, there is a fixed cost (price) for a sybil to purchase the product. Therefore, to enforce sybil-proofness, the mechanisms try to cap the referral fees. 
Douceur and Moscibroda~\cite{Douceur2007} gave a sybil-proof lottery tree mechanism for motivating people to install and run a distributed service in a peer-to-peer system. 
The issue of sybil-attacks have also appeared in many other contexts such as reputation mechanisms~\cite{Cheng2005}, combinatorial auctions~\cite{Todo2009}, social choice~\cite{Wagman2008,Conitzer2010}, and cost-sharing games~\cite{Penna2009}. One major difference between our problem (as well as~\citeN{Babaioff2012}) with these problems is they are dealing with static configurations, which make sybil-proofness hard to achieve. In our results, we punish sybils by reducing their probability of winning in a probabilistic environment. In particular, our mechanism is not sybil-proof if the agents know the outcome of the environment. Such trade-off for sybils, i.e., more reward conditioned on winning with smaller winning probability,  is not available with static inputs. Nevertheless, we believe the results in this paper may be of interest in other settings.

\section{Query incentive networks: a normal form}

In this section, we describe the problem and the model formally. In
particular, we provide formal definition for the random branching
process for information propagation and the various components that
constitute an incentive mechanism.



\subsection{The branching process}

Following previous works, the underlying network is generated by a
Galton-Watson (GW) branching process on an infinite $d$-ary tree. In the branching process, each node $v$ samples its number of
children $C(v)$ independently according to a given distribution $D$.
Afterwards, $v$ selects $C(v)$ children,  out of
its $d$ children uniformly at random, and connects to them. The final tree $T_r$ in the branching process is the connected component
containing the root. We call an agent $u$ {\em active} if $u\in T_r$
and non-active otherwise.


Formally, given an offspring distribution $D=\{c_i\}_{i=0}^d$, where $\sum_{i=0}^dc_i=1$ and $c_i \ge 0$ is the probability to have $i$ children, define the probability generating function of the branching process as
\begin{equation}\label{eq:gefun}
\Psi(x)=\sum_{i=0}^dc_ix^i.
\end{equation}

A basic parameter for the branching process is the branching factor $b=\sum_{i=0}^d i c_i$,
defined as the expectation of $D$.
The extinction probability $\zeta$ of the branching process is the
probability that the branching process dies out, i.e., the final tree
$T_r$ is finite. A well known fact is that $\zeta=1$ if and only if
$b<1$ or $b=1$ with $c_0>0$, and $0\leq
\zeta<1$ otherwise.


For the query issued from the root, each node in the underlying tree has an answer with an independent probability $p=1/n$, where $n$ represents the
rarity of the answer.
So in expectation, we should reach out $\Theta(n)$ nodes to obtain an
answer with constant probability.

\subsection{Query model}


Given the above process for generating the network, the query process works as follows:
\begin{enumerate}
\item The root announces the incentive mechanism, which stipulates rules for selecting the answer and for rewarding the involved agents.
\item The query is propagated, starting from the root, down the tree as generated by the above branching process.
\item Each node (agent) in the tree, when receiving the query, may decide whether to continue to propagate the query, and whether,  in case it has an answer,  to report the answer back to its parent,
\item When an agent has an answer and/or receives reports of answers from its subtree, it chooses to report a subset (can be empty) of the answers to its parent. After the root receives all reports, a winning answer is selected.
\item The holder of the winning answer forwards the answer to the root through the path in the tree.  If any node on the path decides not to forward, no payment will be made. 
\item Once the root receives the answer, the rewards are paid to the nodes according to the rule announced in Step (1).
\end{enumerate}

One significant difference between our model and the previous work is
that in our model the incentive mechanism is determined by the root
but the previous work allows a mixture of global and local contracts.
For example, in the fixed-payment contract or the split contract
schemes, the type of contracts as well as the process for selecting an
answer are fixed globally, but each node can decide locally how to
enter into a contract with its parent or children.




In our model, the incentive mechanism announced by the root contains a
global reward allocation scheme, which maps any final configuration to a set of rewards
to the agents.
While such a global reward scheme may seem limiting as it takes away
the freedom enjoyed by the nodes in the fixed-payment contract and
split contract mechanisms, it is still non-trivial to address the new challenge
of sybil attacks. On the other hand, we show later in
this section that it is possible to describe the equilibrium of a local contract-based
scheme by a global reward allocation scheme.

%
%

\subsection{Query incentive mechanism}

We call an answer (or the agent who has the answer) {\em reachable}
from the root if all the agents on the path from the root to the answer
are active in the branching
process and decide to propagate the query. The {\em query incentive mechanism} determines how an answer
is selected when there are multiple reachable answers, and how the
rewards are allocated.  To prevent arbitrarily complex mechanisms, we
focus on mechanisms that satisfy the following properties.
\begin{enumerate}
\item \emph{Complete}: If there exist reachable answers, the mechanism will select one.
\item \emph{Unique}: Only one answer is selected if multiple answers are presented.
\item \emph{Anonymous}: Agents are not treated preferentially a prior.
\end{enumerate}

Such properties have been implicit in the previous work.  In addition,
by requiring these properties, it makes the analysis easier. For
example, the anonymity is helpful in addressing sybil-proofness since the identities
of the agents can be manipulated.
We
now describe families of mechanisms that achieve the above properties.
We divide the query incentive mechanism into two steps: the answer
selection step and the reward allocation step.

\subsubsection{Answer selection}
\label{sec:answer selection}

The answer selection step chooses one answer when multiple answers are
reachable from the root. We consider two answer selection schemes,
both appeared in the literature, that satisfy the above properties.
\hide{\MP{Li: I commented out the footnote, which seems out of place.}}
\hide{\footnote{In previous work,
\citeN{Cebrian2012} assumed that the answer selection scheme should
not affect the strategies of the agents. Both RW and SP answer
selection schemes satisfy this constraint as observed
in~\cite{Cebrian2012}. However, it is not clear how to mathematically
characterize the set of answer selection schemes which satisfy such
assumption.}}

\begin{itemize}
 \item Random Walk (RW): In RW scheme, starting from the root, at each
 step, we select one child uniformaly at random from those children
 who have reported the existence of answers in its subtree. We
 continue the random walk until we reach an answer, which is selected.

 \item Shortest Path (SP): Among all the reachable answers, we exclude
 those that are not the closest to the root. We then perform the
 above RW process for the remaining answers.
\end{itemize}

Both schemes have very natural interpretation. In the RW scheme,
\emph{after} the formation of the network, each node that has an
answer reports its answer to its parent.  If a node, who does not have
an answer itself, receives reports of multiple answers from all children, it
randomly selects one and reports it to its parent. In particular, if one node has an answer, it will not report answers from its subtree. The process continues until the root selects one of the answers reported to it uniform at random.

The SP scheme on the other hand can be viewed as the RW scheme for
{\em impatient agents}. One node reports back to its parent as soon as
it receives the query, in the case it has an answer. When the node
does not have an answer, it immediately reports back when it receives
reports of answers. In case of multiple answers reported
simultaneously, it will select one uniform at random. However, unlike
RW scheme, the node will not wait the responses from all children.  We
will use the SP answer selection scheme in our mechanisms.\footnote{
The reporting strategy described above is only for interpretation. In a Nash equilibrium, we assume
an agent will report all answers it is aware of. The tie-breaking is handled by the root after receiving all reports.
}


\subsubsection{Reward allocation}

Once an answer is selected, the reward allocation step determines how
the rewards are assigned to the nodes on the {\em answer path} $P$, i.e.,
the path from the root to the selected answer. To achieve anonymity,
we require the reward allocation scheme to be {\em oblivious}.

\begin{definition}{\bf Oblivious Reward Allocation Scheme.}
A reward allocation scheme maps any particular answer path $P$ in a resulting branching process $T_r$ to a set of payments to the agents in $P$:
\[ f: T_r, \, P\in T_r \rightarrow [1,\infty]^{|P|}. \]
The reward scheme $f$ is {\em oblivious} if $\forall P\in T_r, P'\in T_r'$  such that $|P|=|P'|$, we have
\[f(T_r, P) = f(T_r', P').\]
For a node $u\in P$, we denote its reward as $f_u(P)$.
\end{definition}
In particular, an {\em oblivious} reward allocation scheme only cares about the length of the answer
path. It does not consider the identities of the agents
in the path and the structure of the resulting branching process.
Although this is a restriction on the space of reward allocations,
oblivious reward allocation schemes are convenient in case of sybil-attacks,
since both the identities and the structure of the branching process
can be manipulated unexpected. In fact, we will show that the equilibria studied in the fixed-payment contract and split contract based query incentive networks imply {\em oblivious} reward allocation schemes.


{\bf Remark:} In our model, we assume there is no cost for the agents. To avoid trivial reward allocations, we require
all rewards are normalized to be at least one. This is slightly different with previous literature on query
incentive networks where a fixed unit cost is assumed in forwarding
the final selected answer. We choose the reward normalization instead of the forwarding cost to avoid defining the
answer forwarding cost for sybils. Nevertheless, it is
straightforward to generalize our results to the unit forwarding cost case as
long as the cost for sybils is properly defined.

For a query incentive mechanism with an oblivious reward allocation scheme, the \emph{expected cost} of the
mechanism is defined as
\[
\E_{P}[ \sum_{u\in P} f_u(P)  ],
\]
where $P$ is the answer path selected and $f(\cdot)$ is the reward allocation
scheme. If there is no reachable answer, $P$ is empty. The expectation
is taken over the randomness of the branching process, the answer
distribution and the answer selection scheme.

\subsection{Sybil attack}

Once the incentive mechanism is announced, an agent which is reachable
from the root can choose the action to maximize its (expected) payoff.
The agent can choose to propagate the
query or not. When it has answers either from itself or reported from its children, it can choose to report
a (possibly empty) subset of them to its parent.
One important
action we consider is the creation of fake identities (or sybil
attack).  We allow the following type of sybil attacks.


\begin{itemize}
\item {\bf Tree augmentation with sybils.} One agent is allowed to generate a possibly infinite tree of sybils attaching to its parent. Its original children can be attached to {\em one particular sybil} in the tree.\footnote{It would be interesting to consider the case that different children can be attached to different sybils in the tree. Our analysis unfortunately cannot handle this case.}
 The reward of the agent is the total rewards that are received by all its sybils.
\item {\bf Answer placement.} If the agent have an answer, it is allowed to place its answer to any subset of the sybils.
\item {\bf Decision timeframe.} We always assume an agent knows whether it has an answer or not before its strategic decision. We can also assume the action taken by the agent is conditioned on the event that it is {\em active} in the branching process.
\end{itemize}

With the above definition of sybil-attack, we call a query incentive
mechanism is {\em sybil-proof} if the strategy profile in which no
agent generates any sybil is a Nash equilibrium.
\begin{definition} [Sybil-proof query incentive mechanisms]
A query incentive mechanism, consisting of an answer selection scheme and a reward allocation scheme, is sybil-proof with level $h$ if the following strategy profile is a Nash equilibrium:
\begin{itemize}
\item All agents in the underlying $d$-aray tree in level $1$ to $h-1$:  If the agent contains an answer, it directly reports back to its parent and does not propagate the query. Otherwise, the agent chooses to propagate the query to its children. After it receives reports of answers from the children, it reports all answers to its parent. It chooses to forward the answer to its parent if one answer is selected from its subtree.
\item The agents at level $h$: If the agent has an answer, it directly reports back to its parent. If the answer is selected, it chooses to forward the answer to its parent.
The agent will not propagate the query further.
\end{itemize}
\end{definition}

Although we assume the decision of the agent is conditioned on the
event that it is active in the branching process, such condition is in
fact not necessary, i.e., the agent should take the same strategy
independent of whether it is active. This can be easily seen since if the agent is not active, all strategies will have utility zero. In what follows, we will analyze our mechanisms without conditioning on the activeness of the agents.

\hide{Let $u$'s reward be $g(u, A_u)$
where $A_u$ is $u$'s action.  To see this, if agent $u$'s best action
is $A_u$ if $u$ is active in the branching process, $u$ should choose
$A_u$ regardless its activeness. In particular, since when $u$ is not
active, its reward is zero. We have
\begin{align*}
\forall A_u', \E[ g(u, A_u)\, |\, u \mbox{ is active }  ] > \E[ g(u, A_u')\, |\, u \mbox{ is active }  ] \\
\implies \forall A_u', \E[ g(u, A_u)] > \E[ g(u, A_u')] \\
\end{align*}
Therefore, regardless the activeness of $u$ in the branching process, playing $A_u$ will be its best strategy.
In what follows, we will discuss agents' strategies without conditioning on the branching process.

Because our answer selection scheme will only select one {\em answer path} and only cares about the length of the path, it is sufficient to consider sybils formed by a chain, e.g., the agent generates a chain of sybils between its parent and its children. In fact, if the agent does not have an answer, its only chance to be rewarded is its real subtree has an answer and all sybils other than the chain from its parent to the subtree do not help.
On the other hand, if the agent have an answer, by our answer selection schemes, the agent only needs to generate an appropriate length of sybils and place the answer at the end of it. So  it is sufficient for us to only consider sybil-attacks with chains.

An agent's utility is its expected reward collected from itself or its sybils, where the expectation
is taken among the random events forming the branching process and generating answers on tree nodes.}


%
%

\subsection{Formulating contract-based mechanisms by global reward allocation schemes} \label{sec:connection}

To illustrate the connection between the global reward allocation scheme studied
in the paper and the previous work, we describe how one would
formulate an equilibrium of a contract-based mechanism as a global reward allocation scheme. We will mainly discuss the case for the fixed-payment contracts. The case for the split contracts is similar.

In a fixed-payment contract query incentive network,
the root provides a total reward $r_0$ for an answer of the query. In particular, the root enters a contract with its child $u$, i.e., if one answer from $u$ is selected, the root will pay $r_0$ to $u$. Then query is propagated down the tree, during which each node $v$ determines the reward to its children.
Finally, if there are multiple answers reported, the root will select an answer using the RW answer selection scheme. The agents on the path from the root to the selected answer holder are offered the reward based on the fixed-payment contracts determined. In particular, the utility of an agent $v$ in the path is $r_v^p - r_v^c - 1$, where $r_v^p$ is the reward offered by its parent and $r_v^c$ is the reward $v$ offered to its children. If $v$ is holding the selected answer, $r_v^c = 0$. Notice that all agents in the path except the root will pay a unit cost to forward the answer.



The strategy of each participating node $v$ is a function $f_v(\cdot)$, i.e., if the offer from its parent $u$ is $r_u$, it will offer a reward of $r_v = f_v(r_u)$ to its children. Then an equilibrium in the query incentive network is defined by the set of functions $\{f_v\}$ for all participating nodes. (For simplicity, we omit the discussion on the strategy that the nodes decide to participate in this section.)

For any given equilibrium in the fixed-payment contract query incentive network, we can construct a reward allocation scheme as follows. Let $P=\{v_1,v_2,\ldots,v_{\ell} \}$ be the sequence of nodes in the path from the root (excluded) to the selected answer holder. We will reward node $v_i$ by amount of $f_{v_{i-1}}(f_{v_{i-2}}(\cdots f_{v_{1}}(r)))$, where $r$ is the total reward offered by the root. (It $\ell =1$, the first node receives a reward of $r$.) Clearly, this is a valid  reward allocation scheme. Furthermore, if the equilibrium is symmetric, i.e., all nodes in the $i$-th level play the same strategy $f^i(\cdot)$, the corresponding reward allocation scheme is in fact oblivious.

For an equilibrium in the split contract query incentive network, we can construct such a reward allocation scheme as well.
However, it is clear that such an equilibrium does not lead
to a sybil-proof query incentive mechanism.  Consider
the case with a single chain with branching factor $b=1$.  The first
agent who has the answer can create a sybil and sign a split contract with ratio $0$ with it. In this way, the agent gets all the initial reward, which is much more than its fair share in the equilibrium.

For query incentive networks with fixed-payment contracts,
it is shown that there exits a {\em best-interest} Nash equilibrium~\cite{Cebrian2012} 
with a unique strategy function $f(\cdot)$ for all participating agents~\cite{Kleinberg2005}.
Notice that if one agent has an answer, it
will report truthfully and pocket all the reward offered because
creating sybils will not increase the total reward in this case while
it could reduce the probability that its answer is
selected in either RW or SP schemes. In fact, we have the following
result.


\begin{theorem}
\label{thm:fixedcontract}
The {\em best-interest} Nash equilibrium in the fixed-payment contract query incentive network with the RW answer selection scheme
defines a sybil-proof query incentive mechanism.
\end{theorem}
\begin{proof}
Let the best interest Nash equilibrium in the fixed-payment contracts be $\{f(\cdot)\}$ with initial
reward $r_0$. We have shown that if one agent has the answer, it has
no incentive to attack with sybils.
It is sufficient to consider agents that do not have answers.
Now assume the reward allocation scheme
defined by $f(\cdot)$ is not sybil-proof, i.e., one agent at level $i$
has incentive to attack with $k$ {\em additional} sybils between itself and its children in the normal
form.

Correspondingly, consider the original game with fixed-payment contracts. We
show that the same agent at level $i$ would benefit by not following
$f(\cdot)$ in the equilibrium. In particular, in receiving offer
$r_{i-1}$, the agent will offer its children $r_{i+k}=f^{(k+1)}(r_{i-1})$
instead of $r_i = f(r_{i-1})$, where $f^{(k+1)}(\cdot)$ is the function to iteratively apply $f(\cdot)$ for $k+1$ times.

In both cases, i.e., the normal form with $k$ sybils and the fixed-payment
contract game with offer $f^{(k+1)}(r_{i-1})$, all the (real) descendants
will react as if they are lower in the tree.
Notice that the two trees in the two cases are slightly different. However, the RW answer selection scheme will not be impacted by the sybils.
Therefore,
{ for the attacking agent, the probabilities that it is on the answer path are the same in both cases.}

Let $p_1$ and $p_2$ be the probability that the agent is on the answer path if
it does not attack and does attack with $k$ sybils respectively. With
answer selection scheme RW, we have $p_1 \geq p_2$.
By the assumption of the non-sybil-proofness in the normal form without forwarding cost, we have
\begin{equation}
\label{eqn:sybil-proof-normal-form}
 p_2\cdot (r_{i-1} -r_{i+k}) > p_1\cdot (r_{i-1} - r_{i}).
\end{equation}

Notice that in the fixed-payment contract game, the expected utility if this agent plays truthfully is $p_1\cdot (r_{i-1} - r_{i} -1)$. The expected utility if this agent plays $r_{i+k}$ is $p_2\cdot(r_{i-1} - r_{i+k} -1)$. Since $p_2 \leq p_1$,
by Eqn.~(\ref{eqn:sybil-proof-normal-form}),
we have
\[ p_2\cdot (r_{i-1} -r_{i+k} -1) >
p_1 \cdot (r_{i-1} - r_{i}) - p_2 \ge
p_1 \cdot (r_{i-1} - r_{i} -1).\]
In other words, $f(\cdot)$ is not an equilibrium of the game with fixed-payment contracts. It is a contraction.
\end{proof}

{\bf Remark:} We showed the sybil-proofness with the RW answer selection scheme. Our proof does not directly apply to the SP answer selection scheme, which we leave as an interesting problem.

Although fixed-payment contracts already imply a sybil-proof query
incentive mechanism, it is not cost-effective for the case of
$b<2$~\cite{Kleinberg2005} or if one wants to find an answer with
probability $1-\zeta-\frac{1}{n}$~\cite{Arcaute2007}.  In the rest of
the paper, we will propose a new mechanism that is more
cost-effective in these two cases.

\section{Technical preparations}

One main complexity of our analysis comes from the branching
process. In this section, we summarize and develop some technical
results regarding the branching process which will be needed in the
analysis.

Define $\phi_i$ as the probability that there is no answer in the
nodes of the first $i$ levels of the branching process.  Then the
probability that the first answer in our branching process is at level
$i\ge 1$ is $\lambda_i=\phi_{i-1} - \phi_{i}$, with $\phi_0=1$.

Some crucial properties of the sequence of $\{\lambda_i\}$ are
summarized in Lemma~\ref{lem:lambda_charaterization}.  One essential
property is that $\{\lambda_i\}$ is single-peaked: it first increases
approximately geometrically with a constant ratio.  Then it stays at
nearly maximum value for a constant number of levels until it starts
to decrease geometrically. This property might be of independent
interest.

\begin{lemma}
\label{lem:lambda_charaterization}
Consider a branching process with branching factor $b>1$. There exist levels $1\leq \ell^1 \leq \ell^*$ such that
\begin{enumerate}
\item $\{\lambda_i\}$ is a single-peaked sequence, peaking at level $\ell^*$, i.e., $\forall i\le \ell^*$, $\lambda_{i-1} \leq \lambda_{i}$ and $\forall j \ge \ell^*$, $\lambda_{j} > \lambda_{j+1}$.
\item There exists constant $\rho >1$, such that $\forall i <\ell^1$, $\lambda_{i+1} \ge \rho \cdot \lambda_{i}$.
\item $\lambda_{\ell^1} = \Omega(1)$ and consequently $\ell^*-\ell^1 = O(1)$.
\item There exists constant $\gamma >1$, such that $\forall i \ge \ell^*+2$, $\sum_{j\ge i} \lambda_j \leq \gamma  \cdot \lambda_i$.
\end{enumerate}
\end{lemma}

Following the analysis in the literature, we fix the branching factor
$b$ as a constant but we allow $n$ to grow.

Define function
\[t(x) = \sum_{j=0}^dc_jx^j(1-\frac{1}{n})^j.\]
Notice that $\phi_i = t(\phi_{i-1}) = t(t(\phi_{i-2})) = t^{(i)}(1)$.
{This is because, if the root has $j$ children, then the event that there is no answer at the first
	$i$ levels (with probability $\phi_i$) is the same as the event that none of its $j$ children
	has the answer (with probability $(1-\frac{1}{n})^j$) and none of the $i-1$ level subtrees rooted at its children
	has the answer (with probability $\phi_{i-1}^j$).}
The following result studies the growth rate of $\lambda_i$.

\begin{proposition} For all $i\ge 1$,
\label{prop:lambda_ratio}
\begin{equation}
\frac{\lambda_{i+1}}{\lambda_i} \in [ t'(\phi_{i}), t'(\phi_{i-1})].
\end{equation}
\end{proposition}
\begin{proof}
Notice that $\phi(i) = t(\phi_{i-1})$. Therefore, $\lambda_{i+1} = \phi_{i} - \phi_{i+1} = t(\phi_{i-1}) - t(\phi_i).$ As $t(\cdot)$ is continuous, by mean value theorem, there exists $y\in [\phi_i, \phi_{i-1}]$ such that $\lambda_{i+1} = t'(y)(\phi_{i-1} -\phi_i) = t'(y)\cdot \lambda_{i}$. Finally, the result comes by the fact that $t'(x)$ is a monotonically increasing function for $x>0$.
\end{proof}

%
%


We  will also utilize the following two results from previous works.

\begin{lemma}[\cite{Arcaute2007}]
\label{le:t'x}
Given constant $\epsilon>0$, $\forall x\in [1-\epsilon, 1]$,  we have $t{'}(x)\in[(1-\frac{1}{n})\cdot b\cdot (1-5\epsilon d),(1-\frac{1}{n})\cdot b]$.
\end{lemma}

\begin{lemma}[\cite{Cebrian2012}]
\label{lem:lambdai+1}
Consider any GW branching process with branching factor $b>1$. Then, for every $i$ such that $\zeta < \phi_i\leq1$, it holds that
\begin{equation}
\frac{1-\phi_i}{\lambda_{i+1}}\leq \max\{\frac{1}{b-1},\frac{1}{\phi_i - \zeta}\cdot\frac{1}{1-\Psi'(\zeta)}\},
\end{equation}
where $\Psi(x)$ is the generation function of the branching process and $\zeta$ is the extinction probability.
\end{lemma}

Now we are ready to prove the main result in this section.

\begin{proof}[of Lemma~\ref{lem:lambda_charaterization}]
When $b>1$, notice that $\{\phi_i\}$ is a strictly decreasing sequence while $t'(x)$ is increasing for $x>0$. Proposition~\ref{prop:lambda_ratio} indicates that the ratio of $\lambda_{i+1}/\lambda_i$ is decreasing, which implies that $\{\lambda_i\}$ is a single-peaked sequence. This proves (1) of Lemma~\ref{lem:lambda_charaterization}.

Now we proceed to the second property.
Let $\epsilon \in (0,1)$ be some constant we will specify later. Define $\ell(\epsilon) = \max\{ i \,\,| \,\,\phi_i \ge 1-\epsilon \}$.

By Proposition~\ref{prop:lambda_ratio}, we have $\frac{\lambda_{i+1}}{\lambda_i} \ge t'(\phi_i)$.
Let
\begin{equation}
\label{eqn:epsilon1}
\rho = (1-\frac{1}{n})\cdot b\cdot (1-5\epsilon d) >1,
\end{equation}
which holds for sufficiently small $\epsilon$.
By Lemma~\ref{le:t'x} and the definition of $\ell(\epsilon)$, we have for any $1\leq i \leq \ell(\epsilon), \frac{\lambda_{i+1}}{\lambda_i}\ge t'(\phi_i)\ge \rho$.
Therefore, property (2) holds by setting $\ell^1 = \ell(\epsilon)-1$.


For the third property, we show that $\lambda_{\ell(\epsilon)+1}= \Omega(1)$
for some carefully chosen $\epsilon$.
Since $\frac{\lambda_{i+1}}{\lambda_{i}}\leq t'(1)\leq b$, $\lambda_{\ell^1} \geq \lambda_{\ell(\epsilon)+1}/b^2$.
Let $\zeta$ be the extinction probability of the branching process.
Assume
\begin{equation}
\label{eqn:epsilon2}
0< \epsilon < 1-\zeta.
\end{equation}
%
%
%
%
%
By definition of $\ell(\epsilon)$, we have $\phi_{\ell(\epsilon)}\geq 1-\epsilon$ and $\phi_{\ell(\epsilon)+1}<1-\epsilon$. Define $c(\epsilon) = \max\{\frac{1}{b-1},\frac{1}{1-\zeta-\epsilon}\cdot\frac{1}{1-\Psi'(\zeta)}\}$. Notice that $c(\epsilon)$ is non-decreasing for $\epsilon \in (0, 1-\zeta)$.
By Lemma~\ref{lem:lambdai+1}, we have

\begin{equation}
\begin{aligned}
\frac{1-\phi_{\ell(\epsilon)}}{\lambda_{\ell(\epsilon)+1}}&\leq 
c(1-\phi_{\ell(\epsilon)}) \leq c(\epsilon).
\end{aligned}
\end{equation}
Since $\lambda_{\ell(\epsilon)+1} =    \phi_{\ell(\epsilon)} - \phi_{\ell(\epsilon)+1}$, we have

\begin{align*}
\lambda_{\ell(\epsilon)+1} \ge \frac{1-\phi_{\ell(\epsilon)}}{c(\epsilon)}
= \frac{1-\phi_{\ell(\epsilon)+1} - \lambda_{\ell(\epsilon)+1}}{c(\epsilon)} > \frac{\epsilon -  \lambda_{\ell(\epsilon)+1} }{c(\epsilon)}.
\end{align*}
Therefore, $\lambda_{\ell(\epsilon)+1} > \frac{\epsilon}{c(\epsilon)+1} = \Omega(1)$, as both $\epsilon$ and $c(\epsilon)$ are constants independent of $n$.
In other words, for any constant $\epsilon$ satisfies Eqn.~(\ref{eqn:epsilon1}) and Eqn.~(\ref{eqn:epsilon2}), we can set $\ell^1 = \ell(\epsilon)-1$ and both property (2) and (3) hold. Since $\lambda_i$ is growing from $\ell^1$ to $\ell^*$, we must have $\ell^* -\ell^1 = O(1)$. (The total probability is at most $1$.)

Finally, we consider property (4) regarding the sequence of $\{\lambda_i\}$ after $\ell^*$.
By the definition of $\ell^*$ and Proposition~\ref{prop:lambda_ratio},
\begin{equation}
\label{eqn:t'}
 t'(\phi_{\ell^*+1}) \leq \frac{\lambda_{\ell^*+2}}{\lambda_{\ell^*+1}} \leq   t'(\phi_{\ell^*}) \leq  \frac{\lambda_{\ell^*+1}}{\lambda_{\ell^*}} < 1.
\end{equation}
If we can assume $t'(\phi_{\ell^*})$ is a {\em constant} less than $1$, we are done.
However, it is not clear whether $t'(\phi_{\ell^*})$ is bounded away from $1$ or not. (For example, if $t'(\phi_{\ell^*}) = 1-1/n$, we cannot directly conclude on property (4). )
We consider two cases:

(1) $\frac{\lambda_{\ell^*+2}}{\lambda_{\ell^*+1}} < 1/2$. This case is trivial, since for $i \ge \ell^* + 1$, $\lambda_{i+1} < \lambda_i/2$ by the monotonicity of $t'(\cdot)$ and Proposition~\ref{prop:lambda_ratio}. Therefore, property (4) holds with $\gamma =2$.

(2) $\frac{\lambda_{\ell^*+2}}{\lambda_{\ell^*+1}} \ge 1/2$. This case implies $\lambda_{\ell^*+2} \ge \lambda_{\ell^*+1}/2 \ge \lambda_{\ell^*}/4 = \Omega(1)$. Therefore, $\phi_{\ell^*+1} \ge \lambda_{\ell^*+2}  = \Omega(1)$.

Notice that for any {\em constant} $x\in (0,1)$, $t''(x) \ge x^d \sum_{i\ge 2}c_i = \Omega(1)$. ($b>1$ implies $\sum_{i\ge 2}c_i >0$.) Hence, in case (2), we have $t''(\phi_{\ell^*+1}) = \Omega(1)$.

Because $t'(\cdot)$ is continuous, by the mean value theorem, there exists $y \in [\phi_{\ell^*+1},  \phi_{\ell^*}]$, such that
\[ t'(\phi_{\ell^*}) - t'(\phi_{\ell^*+1}) = t''(y)\cdot (\phi_{\ell^*} - \phi_{\ell^*+1})
\ge t''(\phi_{\ell^*+1})\cdot \lambda_{\ell^*+1} = \Omega(1).
\]
Now since $t'(\phi_{\ell^*})<1$ by Eqn.~(\ref{eqn:t'}), we can conclude $\gamma' = t'(\phi_{\ell^*+1})$ is bounded away from $1$, i.e., $\gamma'$ is a constant smaller than $1$. Then for any $i \ge \ell^*+2$, we have $\lambda_{i+1} < \gamma' \cdot \lambda_{i}$. By setting $\gamma = \frac{1}{1-\gamma^{'}}$, we prove property (4).
\end{proof}

%
%

\section{Optimal sybil-proof DR mechanism on chains}

In this section, we discuss the case that the underlying tree is
simply an infinite chain. Each node in the chain has an answer with an
independent probability $1/n$. We design a direct referral mechanism
which is sybil-proof in this case. We also show that, 
in fact,
the DR mechanism is optimal, up to some mild
assumptions.

\subsection{The direct referral reward scheme}
\label{sec:dr_chain}

Let $h$ be the level of the chain that we want to propagate the query to.
Notice that any oblivious reward scheme can be written as a function $r(i,s)$ for $i\ge 1$, $s\ge 0$ and $i+s \le h$, where $r(i,s)$ is the reward to the $i$-th agent on an answer path with length $i+s\le h$, i.e. the first answer appears at level $i+s$. Since we discuss DR mechanisms, we have $r(i,s) = 1$ for any $s\notin \{0,1\}$ and $i+s \leq h$.

Let $p = 1/n$ be the probability that one agent has an answer.
Define $R_{i}=\sum_{s=1}^{h-i}r(i,s)p(1-p)^{s-1}$, i.e., $R_i$ is the expected reward of the $i$-th agent  conditioned on the event that the first $i$ agents do not have any answer.
Let $P_i = \sum_{j=1}^ip(1-p)^{j-1}$ be the probability that there is an answer in $i$ consecutive nodes.
We consider the following reward allocation scheme in the direct referral mechanism.

\begin{definition}
{\bf DR reward allocation scheme on chains.}
Define $r(i,s)$ as:
\begin{align}\label{eq:xichain}
r(i,s)=
\begin{cases}n\cdot {R_{i+1}}
+P_{h-i-1}.
& \text{if $i\leq h-1 \land s=1$,}
\\
\sum_{t=i}^{h-1}r(t,1)+1,  & \text{if $1\leq i\le h \land s=0$,}\\
1 
&\text{$i+s\leq h \land  s>1$,}\\
0, & \text{otherwise.}
\end{cases}
\end{align}
\end{definition}

Notice that $R_i = p\cdot r(i,1) + (1-p)P_{h-i-1}$. Therefore, by definition of $r(i,1)$, we have
\begin{equation}
\label{eqn:R_i}
R_i = R_{i+1} + P_{h-i-1}.
\end{equation}

For what follows, we show that several properties of the direct referral mechanism (a) it is sybil-proof; (b) the expected cost is $O(h^2n)$; and (c) it is {\em the optimal sybil-proof} mechanism, i.e., with the optimal cost.

To show the sybil-proofness, we only need to prove that the following strategy profile is a Nash equilibrium: for each node $v$ with distance less than $h$ from the root, if $v$ does not have an answer, the strategy of $v$ is to propagate the query and does not create sybils; if it has an answer, the strategy of $v$ is to report the answer without creating sybils. For the node at distance $h$, the node will return the answer if it has one; otherwise, it will not propagate further and will not create sybils.

\begin{lemma}
The DR mechanism with above defined reward allocation scheme is sybil-proof on chains.
\end{lemma}
\begin{proof}
For the node $v$ with distance smaller than $h$ from the root, if it does not hold
an answer, the expected reward for propagating the query is always
larger than 0. Thus, it will choose to propagate. If $v$ has an
answer, $v$ does not need to propagate the query, since it can fake
any possible path from itself to a node.  Furthermore, since the DR scheme
only rewards nodes of the first $h$ levels, the node at depth $h$ from the root
will not propagate the query.

Based on the above discussion, for the sybil-proofness, we only need to rule out the strategy that one node create sybils.
For the node at distance $h$ from the root, it is not beneficial to generate fake nodes as the scheme only rewards the first $h$ nodes in the chain. So it is sufficient to consider the nodes with distance less than $h$ to the root.

Consider an agent $v$ with distance $i<h$ to the root. We first assume that $v$ does not have an answer. Denote $D(i,k)$ as the expected reward of $v$ if $v$ duplicates $k\leq h-i$ {\em additional} sybils conditioned on the event that no node from the root to $v$ has an answer.
Then we have
\begin{align*}
D(i,k) &= R_{i+k} + k\cdot P_{h-i-k}
 = R_{i+k-1} - P_{h-i-k} + k\cdot P_{h-i-k}\\
& \le R_{i+k-1} + (k-1)\cdot P_{h-i-k+1}
= D(i,k-1).
\end{align*}
The second equality is by Eqn.~(\ref{eqn:R_i}).
One can then inductively show $D(i,k) \leq D(i,0)$. In other words, node $v$ will not benefit by generating sybils if it does not have an answer.



Next we consider the case that $v$ has an answer and all previous
nodes do not, i.e., $v$ is the first answer holder.  Assume that $v$
generates $k$ sybils, where $k\leq h-i$. Since the DR scheme only
rewards the first answer, $v$ will place the answer at the last sybil. (Otherwise, it may just generate less sybils.) Then the reward $v$ gets in this case is
\begin{align*}
&\sum_{s=0}^{k-1}r(i+s,k-s)+r(i+k,0)\leq\sum_{s=0}^{k-1}r(i+s,1)+r(i+k,0)\\
=&\sum_{s=0}^{k-1}r(i+s,1)+\sum_{s=i+k}^{h-1}r(s,1)+1 =\sum_{t=i}^{h-1}r(t,1)+1 = r(i,0)
\end{align*}
The inequality comes from the fact that $r(i,1)\geq 1$. By the above
inequality, $v$ gets the highest reward without duplication.
Combining all above together, we have shown that the DR mechanism is
sybil-proof.
\end{proof}

\subsection{Efficiency of the DR mechanism}
Next we upper bound the expected total reward of the DR mechanism on
chains. We first compute the values of $r(i,1)$ in the following
lemma.
\begin{proposition}
For $i\le h-1$, $r(i,1) \leq nhP_h+ 1, \mbox{ and } r(i,0) \leq nh^2P_h + h$.
\end{proposition}
\begin{proof}
By Eqn.~(\ref{eqn:R_i}),
$R_i \leq R_{i+1} + P_h$,
which implies $R_i \leq hP_h$. By definition, $r(i,1) = n R_{i+1} + P_{h-i-1} \leq nhP_h+1$. The result on $r(i,0)$ directly follows.
\end{proof}


\begin{lemma}
The expected cost of the DR scheme is $\Theta(P_h^2nh^2)$.
\end{lemma}
\begin{proof}
We first consider the upper bound. Clearly, the cost is dominated by the $r(i,0)$s, which is at most
$P_h\cdot \max_i\{r(i,0)\}= 
O(P_h^2nh^2)$.

Now consider the lower bound. By Eqn.~(\ref{eqn:R_i}), for $i\leq h/4$, $R_i = \sum_{j=1}^{h-i-1}P_j\geq \frac{h}{4}\cdot P_{h/2}$.

By definition of $r(\cdot)$, for $i\leq h/4-1$, $r(i,1) \ge n R_{i+1} \ge  \frac{nh}{4}\cdot P_{h/2}$.
Then for all $i \leq h/8$, $r(i,0) \ge \sum_{j=i}^{h/4} r(j,1) \ge \frac{h}{8}\cdot \frac{nh}{4} P_{h/2}$.
Therefore, the expected cost is at least
\[ P_{h/8}\cdot \min_{i\le h/8}\{r(i,0)\} \ge \frac{nh^2}{32}P_{h/8}\cdot P_{h/2}.\]
The lower bound is attained since $P_{h/2}\geq P_h/2$ and $P_{h/8}\ge P_h/8$.
\end{proof}

We have the main result of this section.

{\sc Theorem~\ref{thm:chain} (restated)}
{\em 	\thmChain
}

\subsection{Optimality of the DR mechanism on chains}\label{subsec:opt}

In fact, our DR mechanism is {\em optimal} with respect to a general class of query incentive mechanisms on chains. To proceed with our discussion, we focus on the query incentive mechanisms with following properties:
\begin{enumerate}
\item The answer selection is deterministic on the chain. (Notice that both RW and SP are deterministic and identical on a chain.)
\item The mechanism is sybil-proof.
\item The mechanism must retrieve an answer if there is one in the first $h$ agents.
\item The reward allocation scheme is normalized, i.e., the reward to each node from the root to the first answer must be at least $1$.
\end{enumerate}
We call a query incentive network is {\em regular} if all four constraints are satisfied.
Notice that our DR mechanism is {\em regular}. Our DR mechanism is in fact has the smallest cost in all resulting configurations among all {\em regular} query incentive mechanisms. 
We provide a proof in Appendix~\ref{section:chainopt}.



\begin{theorem}
\label{thm:chainopt}
The DR mechanism on chains is an optimal regular query incentive mechanism.
\end{theorem}

\section{Sybil-proof DR mechanisms on arbitrary branching processes}

In this section, we present a sybil-proof query incentive mechanism
for an arbitrary branching processes with branching factor $b>1$.  In
particular, we present a Direct Referral (DR) 
mechanism
for any interested height $h$. We will show that the DR mechanism,
which uses the SP answer selection scheme, is sybil-proof. After that,
we show that the expected cost for the DR scheme is $O(h^2)$. It is desirable to point out that compared with previous works, our
analysis does not depend on a particular height $h$, and it does not
require $b\geq 2$.

\subsection{Direct referral query incentive mechanisms with height $h$}

The DR mechanism will be using the SP answer selection scheme. Similar
to the chain case, the reward allocation scheme in the DR mechanism is
oblivious. In particular, the reward allocation scheme can be described by a
function
$r(\cdot,\cdot)$, where $r(i,s)$ is the reward of the $i$-th agent in the selected answer path and the selected answer holder is the $i+s$-th agent for $i\ge 1$ and $s\ge 0$.


{To simply the notation, we decompose the {\em referral reward} (with $s\ge 1$) $r(i,s)$ into two parts:
 $r(i,s) = x(i,s) + y(i,s)$.
$x(i,s)\ge 0$ is for the direct referral and $y(i,s)\ge 0$ ensures that the reward is at least 1.
Specifically, the referral rewards in our the DR mechanism for height $h$ with $i \leq h-1$ and $s\ge 1$ is defined as
$r(i,s)=x(i,s)+y(i,s)$:
}
\begin{align}
x(i,s)=
\begin{cases}\displaystyle{\max_{i+1\le  j\le h}}\left\{x(j,1)+ \frac{j-i}{\lambda_{i+1}}\sum_{\ell = i+1}^h\lambda_\ell \right\}, & \text{if $i\leq h-1$} \text{ and $s=1$,}
\\
0,  & \text{otherwise.}
\end{cases}
\label{eqn:xis}
\end{align}

and
\begin{equation}
y(i,s)=
\begin{cases} 1 & \text{if $i+s\leq h$ and $s > 1$,}
\\
0,  & \text{otherwise.}
\end{cases}
\end{equation}
Finally, the reward of the selected answer holder at level $i$ is defined $a_i = r(i,0)$ for $i\leq h$ as follows. $a_h = 1$ and for $i<h$,

\begin{equation}
\label{eqn:tree:a_i}
a_i = x(i,1) + a_{i+1} + 1 \Rightarrow a_i = (h-i+1)+ \sum_{i\leq j\leq h} x(j,1).
\end{equation}

Informally, in the above reward scheme, if the selected answer is at level
$i\leq h$, the answer holder receives $a_i$, its direct referral at
level $i-1$ receives $x_{i-1}+1$ and all other agents from the root to
the selected answer holder receive $1$.
For simplicity, we abbreviate the sequence $\{x(i,1)\}$ as $\{x_i\}$.

{\bf Remark:} Our mechanism requires a specified height $h$. Furthermore, to compute the rewards, we need the knowledge of the branching process. Notice that, such knowledge is also required in previous literature to compute the initial reward as well as the contracts between the agents.


\subsection{Sybil-proofness of the DR query incentive mechanism}

We show that the DR  mechanism defined above is
sybil-proof. Eqn.~(\ref{eqn:tree:a_i}) implies that if one agent at
level $i$ has an answer, it does not have incentive to generate
sybils, i.e., $a(i)$ is its largest possible reward and generating
sybils will not increase its chance to be selected.
So to show the sybil-proofness of the DR mechanism, we only need to argue for those agents who do not have answers. The following lemma characterizes the probabilities that one agent is rewarded when it does not have an answer.

\def\Rev{\mathrm{Rev}}
\def\DR{\mathrm{DR}}
\def\NA{\mathrm{NA}}


\begin{proposition}
\label{prop:drprob}
For an agent $v$ at level $i<h$ in the underlying $d$-ary tree of the branching process,
let $\Rev(v)$ be the event that $v$ is on the path from the root to the selected answer holder $u\neq v$ and $\DR(v)$ be the event that $u$ is a child of $v$.  We have
\[
\pr[\Rev(v)] = \frac{\sum_{j=i+1}^h\lambda_i}{d^{i}} \mbox{ and }
\pr[\DR(v)] = \frac{\lambda_{i+1}}{d^{i}}.
\]
Let $\NA(v)$ be the event that $v$ does not have an answer.
\[
\pr[\Rev(v)\,|\, \NA(v)] = \frac{n}{n-1}\cdot \frac{\sum_{j=i+1}^h\lambda_i}{d^{i}} \mbox{ and }
\pr[\DR(v)\,|\, \NA(v)] = \frac{n}{n-1}\cdot\frac{\lambda_{i+1}}{d^{i}}.
\]
\end{proposition}
\begin{proof}
$\Rev(v)$ happens if the selected answer $u$ is at level $i+1$ or
lower, which is with probability $\sum_{j=i+1}^h \lambda_j$ as the
mechanism only retrieves an answer in the first $h$ levels. Notice
that there are exactly $d^{i}$ agents at level $i$. (The root is
treated as level $0$.) By symmetry, the probability that $u$ is in the
subtree of $v$ is exactly $\frac{1}{d^{i}}$. Therefore, $\pr[\Rev(v)]
= \sum_{j=i+1}^h \lambda_j/{d^{i}}$. By the same argument, we have
$\pr[\DR(v)] = {\lambda_{i+1}}/{d^{i}}$.

Now consider the two probabilities conditioned on the event $\NA(v)$ with $\pr[\NA(v)] = (1-1/n)$. Notice that
$\Rev(v)$ (resp. $\DR(v)$) implies $\NA(v)$. In other words, $\pr[\Rev(v) \land \NA(v)] = \pr[\Rev(v)]$ (resp.
$\pr[\DR(v) \land \NA(v)] = \pr[\DR(v)]$). The results then come from Bayes' theorem.
\end{proof}

\begin{lemma}
The above DR mechanism is sybil-proof.
\end{lemma}
\begin{proof}
As discussed earlier, it is sufficient to consider agents that do not have answers.
Consider such an agent $v$ at level $i\leq h-1$.
Suppose $v$ will get more reward by generating $j$ sybils and attach its original subtree to the last sybil.


Consider the case that $v$ does not create sybils.
Conditioning on the even $\NA(v)$, i.e. $v$ does not have an answer,
by Proposition~\ref{prop:drprob},  the probability that $v$ receives the  direct referral fee is $\frac{\lambda_{i+1}}{d^{i}}\cdot \frac{n}{n-1}$ and the probability that the answer selected is in the subtree rooted at $v$ is $\sum_{j=i+1}^h \frac{\lambda_j}{d^{i}}\cdot
\frac{n}{n-1}$.

Notice that if $v$ generates sybils, both probabilities will decrease. This can be shown by coupling the randomness in both cases. For any fixed realization of the branching process and the answer placement, the probability that a particular answer in $v$'s subtree is selected is always higher (or equal) in the case that $v$ does not generate sybils by the SP answer selection scheme.

Now we consider the rewards in both cases. If $v$ does not generate sybils, the expected reward is
\begin{equation}
\label{eqn:rewardwithoutattack}
R_i^0 = \frac{n}{n-1}\cdot\left(\frac{\lambda_{i+1}}{d^{i}} \cdot x_i + \frac{\sum_{j=i+1}^h\lambda_j}{d^{i}}\right).
\end{equation}
The first part is the reward from direct referral $x_i$ and the second part comes from $y(i,s)$.
For the case that $v$ generates $j\ge 1$ sybils, the expected reward of $v$ and its sybils is at most
\begin{equation}
\label{eqn:rewardwithattack}
R_i^j\leq\frac{n}{n-1}\cdot\left(\frac{\lambda_{i+1}}{d^{i}} \cdot x_{i+j} + (j+1)\cdot \frac{\sum_{j=i+1}^h\lambda_j}{d^{i}}\right).
\end{equation}
By Eqn.~(\ref{eqn:xis}), $x_i \ge x_{i+j} + j\cdot\frac{\sum_{\ell=i+1}^h\lambda_\ell}{\lambda_{i+1}}$, which implies $R_i^0 \ge R_i^j$. This is a contradiction.
Therefore, agent $v$ will not benefit by generating sybils if it does not have the answer. Hence the DR mechanism is sybil-proof.
\end{proof}

%
%
%

\subsection{The expected cost of the DR mechanism}

Now we start analyzing the expected cost of the DR mechanism, which can be described as follows
\begin{align}
&\sum_{i=1}^h \lambda_i \cdot a_i+ \sum_{i=2}^h \lambda_i \cdot x_{i-1}+ \sum_{i=1}^h \lambda_i \cdot (i-1)
\label{eqn:total_cost}
\end{align}
The first term is the reward to the answer holder. The second term is the reward for the direct referral and the third term is for all other agents that forwarded the answer.
Before we analyze the cost in Eqn.~(\ref{eqn:total_cost}), we characterize the sequence $\{x_i\}$.

\begin{observation}
The sequence of $\{x_i\}$ is decreasing.
\end{observation}

\begin{proposition}\label{prop:xki}
For $i\ge \ell^*+1$, $x_i\le \gamma\cdot (h-i)$, where $\gamma$ is the constant in Lemma~\ref{lem:lambda_charaterization}.
\end{proposition}
\begin{proof}
We prove the statement by induction. The statement holds for $i=h$ as $x_h=0$. Suppose it holds for all $i \ge j$. Now consider $j-1 \ge \ell^*+1$. By construction,
\begin{equation*}
\begin{aligned}
x_{j-1} = \max_{\ell \ge j}\{ x_\ell + \frac{(\ell-j+1)}{\lambda_{j}}\sum_{s = j}^h\lambda_s \}
 	   \le \max_{\ell \ge j} \{ \gamma\cdot(h-\ell) + \gamma\cdot(\ell -j+1)                    \}
 	   = \gamma \cdot(h-j+1).
\end{aligned}
\end{equation*}
The inequality is by Lemma~\ref{lem:lambda_charaterization} (4).
\end{proof}

In other words, we will not pay a lot of referral fee for agents beyond level $\ell^*$.

\begin{proposition}\label{prop:xik}
For all $i\leq \ell^*$, we have $\lambda_{i+1}\cdot x_i \leq (\gamma+1)\cdot (h-i)$.
\end{proposition}
\begin{proof}
The statement is proved by induction. We first consider the case $i=\ell^{*}$. By definition in Eqn.~(\ref{eqn:xis}),
\begin{align}
x_{\ell^*} &= \max_{j\ge \ell^*+1}\{x(j)+ \frac{j-\ell^*}{\lambda_{\ell^*+1}}\sum_{\ell = \ell^*+1}^h\lambda_\ell \}
\leq\max_{j\ge \ell^*+1}\{x(j)+ (j-\ell^*)\cdot(1+\frac{\gamma\cdot \lambda_{\ell^*+2}}{\lambda_{\ell^*+1}})\}\nonumber\\
&\leq\max_{j\ge \ell^*+1}\{\gamma\cdot(h-j)+ (j-\ell^*)\cdot(1+\gamma)\}
\leq (\gamma+1)\cdot (h-\ell^*). \label{eqn:xellstar}
\end{align}
The first inequality is by Lemma~\ref{lem:lambda_charaterization} (4) and the second inequality is by Lemma~\ref{lem:lambda_charaterization} (1) and Proposition~\ref{prop:xki}.

Now suppose it holds for all $j\leq i\leq \ell^*$. Consider $j-1\ge 1$, we have
\begin{equation}
\begin{aligned}
\lambda_j\cdot x_{j-1} = \max_{\ell \ge j}\{ \lambda_j\cdot x_\ell + (\ell-j+1)\sum_{k= j}^h\lambda_k \}
	   \leq  \max_{\ell \ge j}\{ \lambda_j\cdot x_\ell + (\ell-j+1) \}.
	   \label{eqn:max_ell}
\end{aligned}
\end{equation}
Assume  the term $\max_{\ell \ge j}\{ \lambda_j x_\ell + (\ell-j+1) \}$ is optimized at $\ell = j^*$.
Consider two cases:

\emph{Case 1:} $j^*\ge \ell^*$. In this case,
$\lambda_{j}\cdot x_{j^*} + (j^*-j+1) \leq (\gamma+1)\cdot(h-j^{*})+ (j^{*}-j+1)\leq (\gamma+1)\cdot(h-j+1)$ by Proposition~\ref{prop:xki} and Eqn.~(\ref{eqn:xellstar}).

\emph{Case 2:} $j\le j^* <\ell^*$. We have
\begin{equation*}
\begin{aligned}
\lambda_j\cdot x_{j^*} +(j^*-j+1) \leq \lambda_{j^*+1}x_{j^*} +(j^*-j+1)
&\leq (\gamma+1)\cdot(h-j^*)+(j^*-j+1)\\
&\leq (\gamma+1)\cdot(h-j+1).
\end{aligned}
\end{equation*}
The first inequality comes from $\lambda_j \leq \lambda_{j^*+1}$ by Lemma~\ref{lem:lambda_charaterization} (1).
The second inequality is by induction. Therefore, in both cases, we have
$\lambda_j\cdot x_{j-1} \leq (\gamma+1)(h-j+1)$.
\end{proof}

We are ready to show the expected cost of the above mechanism.

\begin{lemma}
The expected reward of the DR query incentive mechanism is $O(h^2)$.
\end{lemma}

\begin{proof}
By Proposition~\ref{prop:xik} and Proposition~\ref{prop:xki}, the total referral fee is
\begin{align}
\sum_{i=2}^h \lambda_i \cdot x_{i-1} \leq \sum_{i=2}^{\ell^{*}}(\gamma+1)(h-i) + \gamma\cdot (h-\ell^{*}-1)\cdot \sum_{j\ge \ell^*+1}\lambda_j=O(h^2).\label{eqn:totaldrfee}
\end{align}





Now we analyze the total expected reward for answer holders.
\begin{align}
\sum_{i=1}^h \lambda_{i}\cdot a_i = \sum_{i=1}^h \lambda_{i}\cdot (h-i+1 +\sum_{j={i}}^hx_j   )
& \le h + \sum_{j=1}^h x_j \sum_{i=1}^{j} \lambda_i \nonumber\\
&\leq h + \sum_{j=1}^{\ell^1-1} x_j \sum_{i=1}^{j}\lambda_i + \sum_{j=\ell^1}^{\ell^*} x_j + \sum_{j=\ell^*+1}^{h} x_j, \label{eqn:reward_answer_holder}
\end{align}

where $\ell^{1}$ is defined in Lemma~\ref{lem:lambda_charaterization}. We inspect the terms individually. Consider the second term in Eqn.~(\ref{eqn:reward_answer_holder}). By Lemma~\ref{lem:lambda_charaterization} (2), $\{\lambda_i\}$ grows at a rate of $\rho >1$ until $\lambda_{\ell^1}$ for $\ell^1 \le \ell^*$. Therefore, for any $i\leq \ell^1$, $\sum_{j=1}^i \lambda_j \leq \frac{\rho}{\rho-1}\lambda_i$.
Then
\[
\sum_{j=1}^{\ell^1-1} x_j \sum_{i=1}^{j}\lambda_i \leq \sum_{j=1}^{\ell^1-1} x_j \cdot \lambda_{j}\cdot \frac{\rho}{\rho-1}\leq \sum_{j=1}^{\ell^1-1} x_j \cdot \lambda_{j+1}\cdot \frac{\rho}{\rho-1}= O(h^2).
\]
The second inequality is by the fact $\lambda_i\leq \lambda_{i+1}$ for $i\leq \ell^{*}-1$ in Lemma~\ref{lem:lambda_charaterization} (1) and the last equality comes from Eqn.~(\ref{eqn:totaldrfee}).

For the third term, by Proposition~\ref{prop:xik} for any $\ell^1 \leq j \leq \ell^*$,
\[
x_j
\leq (\gamma+1)(h-j)\cdot \frac{1}{\lambda_{j+1}}.
\leq (\gamma+1)(h-j)\cdot \frac{1}{\lambda_{\ell^1}} = O(h).
\]

The last equality comes from $\lambda_{\ell^*}\ge \lambda_j \ge \lambda_{\ell^1} = \Omega(1)$ by Lemma~\ref{lem:lambda_charaterization} (3). Also by Lemma~\ref{lem:lambda_charaterization} (3), $\ell^* - \ell^1 = O(1)$ and the third term in Eqn.~(\ref{eqn:reward_answer_holder}) is $O(h)$.

Finally, by Proposition~\ref{prop:xki}, the last term in Eqn.~(\ref{eqn:reward_answer_holder}) is $O(h^2)$. We conclude that
\begin{equation}
\label{eqn:reward_answer_holder_bound}
\sum_{i=1}^h \lambda_{i}\cdot a_i = O(h^2).
\end{equation}
Combining Eqn.~(\ref{eqn:totaldrfee}) and Eqn.~(\ref{eqn:reward_answer_holder_bound}), the expected cost of Eqn.~(\ref{eqn:total_cost}) is $O(h^2)$.
\end{proof}

We obtain the main result of this paper.

{\sc Theorem~\ref{thm:tree} (restated)}
{\em 	\thmTree
}

\section{Conclusion}

In this paper, we developed a set of sybil-proof query incentive
mechanisms which allocate most rewards to the selected answer holder
as well as its direct referral. We show that such a direct referral
mechanism, with properly designed rewards, is efficient, with expected
cost $O(h^2)$ for any $b>1$. In contrast, all previous mechanisms
require a reward exponential to $h$ when $h$ is large enough for
finding an answer with high probability. We also showed that if the underlying branching process is a chain, our DR mechanism is optimal with some mild assumptions.
Since our reward scheme has very simple structure, it might have good chance to be adopted in practice.

One drawback of our DR mechanism is that it is only efficient in
expectation. In some cases, e.g., when the first agent has an answer,
the cost will be exponential to $h$ even in the case of $b>2$ and $h
=O(\log n)$. As a result, our mechanism is not collusion-proof. Agents
can collude with each other to improve their overall reward.  It is an
interesting open problem to design efficient sybil-proof incentive
mechanisms to address these issues.

\bibliographystyle{acmsmall}
\bibliography{reference}

\newpage
\appendix

\section{The optimality of the DR mechanism on chains}
\label{section:chainopt}



We first characterize such {\em regular } sybil-proof mechanisms.

\begin{lemma}\label{pro:strongchain}
On a chain, if a {\em regular} query incentive mechanism is sybil-proof, its reward scheme will only reward the first answer holder and its ancestors.
\end{lemma}

\begin{proof}
Consider any final answer configuration of a chain with length $h$. If the reward scheme will reward agents beyond the first answer holder,
in our sybil-attacking model, the first answer holder can produce the remaining chain by its sybils and collect all remaining rewards. Notice that the agents can observe the results of its ancestors. Therefore, such reward scheme is not sybil-proof.
%
%
\end{proof}

Since any {\em regular} sybil-proof mechanism only rewards the first answer, we can use the reward framework given in Section~\ref{sec:dr_chain}. In particular, we define
$r(i,s)$ as the reward to the $i$-th agent on an answer path with length $i+s\le h$, i.e., the first answer appears at level $i+s$.

\begin{proof}[of Theorem~\ref{thm:chainopt}]
For simplicity, we define $a(i) = r(i,0)$.
Consider another {\em regular} mechanism defines reward $\hat r(i,s)$ to the $i$-th node if the selected answer holder is at the $i+s$ position and reward $\hat a(j)$ to the answer holder if the answer is at the $j$-th position.
In order to show the optimality of our mechanism, we only need to prove that $\hat r(i,s)\geq r(i,s)$ and $\hat a(j)\geq a(j)$.

We first show that $\hat r(i,s)\geq r(i,s)$. Note that we only need to prove this for $i\leq h-1$ and $s=1$, since in other cases the inequality is obviously correct.

Let $\hat R(i,k)=\sum_{s=1}^{h-i-k}p(1-p)^{s-1}\sum_{j=0}^{k}\hat r(i+j,s+k-j)$, which is the expected reward that the node of depth $i$ can get by creating $k$ {\em additional} sybils conditioned on the event that all previous nodes including itself do not have any answer. We inductively show that $\hat r(i,s)\geq r(i,s)$ for $i\leq h-1$ and $s=1$. This is clearly true for $i=h-1$. Assume that this is true for $i+1$. For $i$, since the other mechanism is sybil-proof in a Nash equilibrium, we have
\begin{align}
\hat R(i,0)\geq \hat R(i,1)
&=\sum_{s=1}^{h-i-1}p(1-p)^{s-1}\hat r(i,s+1)+\sum_{s=1}^{h-i-1}p(1-p)^{s-1}\hat r(i+1,s)\nonumber\\
&\ge \sum_{s=1}^{h-i-1}p(1-p)^{s-1}\hat r(i,s+1)+\sum_{s=1}^{h-i-1}p(1-p)^{s-1}r(i+1,s) \label{eqn:app_induction}\\
&=\sum_{s=1}^{h-i-1}p(1-p)^{s-1}\hat r(i,s+1)+R_{i+1}. \nonumber
\end{align}
The inequality in Eqn.~(\ref{eqn:app_induction}) is by induction.
So we have
\begin{equation}
\begin{aligned}
\hat R(i,0)=p\cdot \hat r(i,1)+\sum_{s=2}^{h-i}p(1-p)^{s-1}\hat r(i,s)
\geq \sum_{s=1}^{h-i-1}p(1-p)^{s-1}\hat r(i,s+1)+R_{i+1}
\end{aligned}
\end{equation}

Rearranging the terms in the above inequality, we have
\begin{align*}
p\cdot \hat r(i,1)&\geq \sum_{s=1}^{h-i-1}p(1-p)^{s-1}\hat r(i,s+1)-\sum_{s=2}^{h-i}p(1-p)^{s-1}\hat r(i,s)+R_{i+1}\\
& = \sum_{s=1}^{h-i-1}p(1-p)^{s-1}\hat r(i,s+1)-\sum_{s=1}^{h-i-1}p(1-p)^{s}\hat r(i,s+1)+R_{i+1}\\
&=p\cdot \sum_{s=1}^{h-i-1}p(1-p)^{s-1}\hat r(i,s+1) + R_{i+1}\\
&\ge p\cdot P_{h-i-1} + R_{i+1} \\
&=p\cdot r(i,1)
\end{align*}
Thus, $\hat r(i,1)\geq r(i,1)$.

Next we will inductively show that $\hat a(j)\geq a(j)$. For $j=h$, $\hat a(j)\geq 1=a(j)$. Assume that $\hat a(j)\geq a(j)$ for all $j\geq i+1$, then for $j=i$, by the sybil-proofness of the other mechanism, we have
\begin{equation}
\begin{aligned}
\hat a(i)\geq \hat r(i,1)+\hat a(i+1)\geq r(i,1)+a(i+1)=a(i)
\end{aligned}
\end{equation}
 Combining all together, we have shown the optimality of our DR mechanism.
\end{proof}

\end{document}